\documentclass[a4,a4wide]{article}

\usepackage{graphicx}
\usepackage{paralist}

\usepackage{ifthen}
\usepackage{url}
\usepackage{amssymb}
\usepackage{amsmath}

\usepackage{xparse}
\usepackage{amsmath}

\usepackage[usenames,dvipsnames]{xcolor}
\usepackage{xspace}
\usepackage{datatool}
\usepackage{glossaries}

\newcommand{\m}[1]{\ensuremath{#1}\xspace}
\newcommand{\trval}[1]{\m{\mbox{\bf #1}}}

	\newcommand{\limplies}{\m{\Rightarrow}}
	
	\newcommand{\lequiv}{\m{\Leftrightarrow}}

	\newcommand{\lrule}{\m{\leftarrow}}
	\newcommand{\cause}{\m{\stackrel{c}{\lrule}}}
	
	\newcommand{\ltrue}{\trval{t}}
	\newcommand{\lfalse}{\trval{f}}
	\newcommand{\lunkn}{\trval{u}}
	
	\newcommand{\voc}{\m{\Sigma}}

	\newcommand{\struct}{\m{I}}
	
	\newcommand{\I}{\m{\mathcal{I}}}
	
	\newcommand{\J}{\m{\mathcal{J}}}

	\newcommand{\theory}{\m{\mathcal{T}}}

	\newcommand{\D}{\m{\Delta}}

	\NewDocumentCommand\inter{g+g}{%
	  \IfNoValueTF{#1}
	    {\struct}
	    {\m{#1^{#2}}}}

	\newcommand{\xxx}{\m{\overline{x}}}
	\newcommand{\yyy}{\m{\overline{y}}}
	
	\newcommand{\ddd}{\m{\overline{d}}}

	\newcommand{\ttt}{\m{\overline{t}}}

	\newcommand{\nat}{\m{\mathbb{N}}}
	
	\renewcommand{\int}{\m{\mathbb{Z}}}

	\newcommand{\leqp}{{\m{\,\leq_p\,}}}
	\newcommand{\geqp}{{\m{\,\geq_p\,}}}

	\NewDocumentCommand\subs{g+g}{%
	  \IfNoValueTF{#1}
	    {\m{/}}
	    {\m{#1/ #2}}}

	\newcommand{\logicname}[1]{\text{\sc #1}\xspace}

	\newcommand{\idp}{\logicname{IDP}}

	\newcommand{\fodot}{\logicname{FO(\ensuremath{\cdot})}}
	
	\newcommand{\foid}{\logicname{FO(\ensuremath{ID})}}
	
	\newcommand{\foidplus}{\logicname{C-Log}} 
	\newcommand{\clog}{\logicname{C-Log}}
	\newcommand{\foclog}{\logicname{FO(\clog)}}

	\newcommand{\fo}{\FO}
	\newcommand{\esoid}{\logicname{\ensuremath{\exists}SO(\ensuremath{ID})}}

\newcommand{\ouracronym}[3]{%
	\newacronym{#1}{#2}{#3}
	\expandafter\newcommand\csname #1\endcsname{\gls{#1}\xspace}%
}
	\ouracronym{FO}{FO}{first-order logic}
	\ouracronym{MX}{MX}{Model Expansion}
	\ouracronym{MO}{MO}{Model Optimization}
	\ouracronym{ASP}{ASP}{Answer Set Programming}
	\ouracronym{TP}{TP}{Theorem Proving}
	\ouracronym{LP}{LP}{Logic Programming}
	\ouracronym{CP}{CP}{Constraint Programming}
	\ouracronym{KR}{KR}{Knowledge Representation}
	\ouracronym{CSP}{CSP}{Constraint Satisfaction Problem}
	\ouracronym{SMT}{SMT}{SAT Modulo Theories}
	\ouracronym{KBS}{KBS}{knowledge-base system}
	\ouracronym{NNF}{NNF}{Negation Normal Form}
	\ouracronym{FNNF}{FNNF}{Flat Negation Normal Form}
	\ouracronym{DefNNF}{DefNNF}{Definition Negation Normal Form}
	\ouracronym{CDCL}{CDCL}{Conflict-Driven Clause-Learning}
	\ouracronym{LCG}{LCG}{Lazy Clause Generation}

	\makeatletter
	\def\ifenv#1{
	\def\@tempa{#1}%
	\def\@ttempa{#1*}%
	\ifx\@tempa\@currenvir
	\expandafter\@firstoftwo
	\else
	\expandafter\@secondoftwo
	\fi
	}
	\makeatother

	\newcommand{\ddrule}[4]{\ensuremath{#1 \leftarrow #2 & \{#3\} & #4}}
	\newcommand{\drule}[2]{\ensuremath{#1 & \leftarrow & #2}}

	\newcommand{\darule}[4]{\ensuremath{#1 \leftarrow #2 & \{#3\} & #4}}
	\newcommand{\arule}[2]{\ensuremath{#1 \, &\leftarrow \, #2}}

	\newenvironment{ldef}{\left\{\begin{array}{l@{ \,}l@{\,}l}}{\end{array}\right\}}
	\newenvironment{ltheo}{\[\begin{array}{l}}{\end{array}\]\ignorespacesafterend}

	\newcommand{\LNDRule}[2]{
	\ifenv{array}
	{\drule{#1}{#2}}
	{ \ifenv{align}
		{\arule{#1}{#2}}
		{\ifenv{align*}
		{\arule{#1}{#2}}
		{ERROR: using LDRule in unsupported environment: \@currenvir}
		}
	}
	}

	\newcommand{\LDRule}[4]{
	\ifenv{array}
	{\ddrule{#1}{#2}{#3}{#4}}
	{ \ifenv{align}
		{\darule{#1}{#2}{#3}{#4}}
		{\ifenv{align*}
		{\darule{#1}{#2}{#3}{#4}}
		{ERROR: using LDRule in unsupported environment: \@currenvir}
		}
	}
	}

	\NewDocumentCommand\LRule{m+g+g+g}{%
		\IfNoValueTF{#2}%
		{#1.&}{%
		\IfNoValueTF{#3}
		{\LNDRule{#1}{#2.}}
		{\LDRule{#1}{#2.}{#3}{#4}}%
		}
	}

	\NewDocumentCommand\CLRule{m+g}{%
	\ifenv{array}
	{\cdrule{#1}{#2}}
	{ \ifenv{align}
		{\carule{#1}{#2}}
		{\ifenv{align*}
			{\carule{#1}{#2}}
			{ERROR: using CLRule in unsupported environment: \@currenvir}
		}
	}
	}

	\NewDocumentCommand\carule{m+g}{%
		\IfNoValueTF{#2}
			{\ensuremath{#1.}}
			{\ensuremath{#1 \, &\cause \, #2}}}
	\NewDocumentCommand\cdrule{m+g}{%
		\IfNoValueTF{#2}
			{\ensuremath{#1.}}
			{\ensuremath{#1 & \cause & #2}}}

	\newcommand{\algrule}[4]{
	\hbox{{#1}:}& 
	\quad #2 ~\longrightarrow~ #3 
	\hbox{~ if } #4\\
	}

	\newcommand{\AlgoRule}[4]{
	\ifenv{array}
	{\algrule{#1}{#2}{#3}{#4}}
		{ERROR: using AlgoRule in unsupported environment: \@currenvir}
	}

\newcommand{\commentstyle}{\color{Gray}}

	\usepackage[final]{listings}

	\lstdefinelanguage{idp}{
		morekeywords=[1]{namespace,vocabulary,theory,structure,procedure,term,set,formula, spec, specification},
		morekeywords=[2]{include,using,type,isa,contains,partial,extern,LFD,GFD,constructed,from,constraint,func,pred,supertype,of,subtype,define},
		morekeywords=[3]{int,float,char,string,nat},
		morekeywords=[4]{if,then,else,for,end},
		morecomment=[s]{/*}{*/},	
		morecomment=[l]{//}
	}
	\newcommand{\ignore}[1]{}

	\newboolean{nocomments}
	\setboolean{nocomments}{false}

	\newboolean{commentmargin}
	\setboolean{commentmargin}{true}

	\newcommand{\namedcomment}[3]{
		\ifthenelse{\boolean{nocomments}}
		{} 
		{ 
			\ifthenelse{\boolean{commentmargin}}
				{ {\color{#3} \marginpar{\color{#3}\sc #2}#1}  } 
				{  {\color{#3} {\sc #2}: #1}  } 
		}
	}
	\newcommand{\mnamedcomment}[3]{\ifthenelse{\boolean{nocomments}}{}{{\marginpar{ \color{#3}{\sc #2}:#1}}}}

	\newcommand{\todo}[1]{\namedcomment{#1}{TODO}{blue}}

	\newcommand{\bart}[1]{\namedcomment{#1}{bb}{red}}
	
	\newcommand{\marc}[1]{\namedcomment{#1}{md}{orange}}

	\usepackage{soul}

\newcommand{\keyword}[2]{%
	\expandafter\newcommand\csname #1\endcsname{#2\xspace}%
	\expandafter\newcommand\csname #1s\endcsname{#2s\xspace}%
	\expandafter\newcommand\csname #1ness\endcsname{#2ness\xspace}%
}
\usepackage{etoolbox}

\newcommand\setcitation[2]{%
  \csdef{mycommoncitation#1}{#2}}
\newcommand\getcitation[1]{%
  \csuse{mycommoncitation#1}}

\setcitation{IDP}{WarrenBook/DeCatBBD14}
\setcitation{fodot}{tocl/DeneckerT08}
\setcitation{CP}{Apt03}
\setcitation{EZCSP}{lpnmr/Balduccini11}
\setcitation{KR}{Baral:2003}
\setcitation{ASPComp2}{lpnmr/DeneckerVBGT09}
\setcitation{ASPComp3}{LPNRM/Calimeri11}
\setcitation{ASPComp4}{conf/lpnmr/AlvianoCCDDIKKOPPRRSSSWX13}
\setcitation{CPSupport}{ictai/DeCat13}
\setcitation{functionDetection}{iclp/DeCatB13}
\setcitation{fodot2asp}{Denecker12}
\setcitation{Inca}{aaai/DrescherW11}
\setcitation{csp2asp}{ijcai/DrescherW11}
\setcitation{DPLLT}{cav/GanzingerHNOT04}
\setcitation{AspInPractice}{synthesis/2012Gebser}
\setcitation{clasp}{ai/GebserKS12}
\setcitation{oclingo}{kr/GebserGKOSS12}
\setcitation{gringo}{lpnmr/GebserST07}
\setcitation{cmodels}{aaai/GiunchigliaLM04}
\setcitation{inputster}{tplp/Jansen13}
\setcitation{DLV}{tocl/LeonePFEGPS06}
\setcitation{LearningPaper}{corr/Blockeel13}
\setcitation{clog}{iclp/BogaertsVDV14}
\setcitation{AFT}{DeneckerBV12}
\setcitation{}{}
\setcitation{}{}
\setcitation{}{}
\setcitation{}{}
\setcitation{}{}
\setcitation{}{}
\setcitation{}{}
\setcitation{}{}
\setcitation{}{}
\setcitation{}{}
\setcitation{}{}
  
\newcommand\mycite[1]{%
      \ifcsname mycommoncitation#1\endcsname%
   \cite{\getcitation{#1}}%
  \else%
    {\color{blue}ERROR: undefined citation key ``#1''}
  \fi%
}	
  
\usepackage{amsthm}

\newcommand{\univBQ}[2]{\m{\forall #1 [#2]}}
\newcommand{\exisBQ}[2]{\m{\exists #1 [#2]}}

\newcommand{\pred}[1]{\m{\mathrm{#1}}}

\NewDocumentCommand\lattice{g+g}{%
  \IfNoValueTF{#1}
    {\m{L}}
    {\m{\left(#1,#2\right)}}%
}

\newcommand{\Bilat}[1]{\m{#1^2}}
\NewDocumentCommand\bilat{g}{%
  \IfNoValueTF{#1}
    {\Bilat{\lattice}}
    {\Bilat{#1}}%
}

\newcommand{\cBilat}[1]{\m{#1^c}}
\NewDocumentCommand\cbilat{g}{%
  \IfNoValueTF{#1}
    {\cBilat{\lattice}}
    {\cBilat{#1}}%
}

\newcommand{\Operator}{\m{O}}
\NewDocumentCommand\Op{g}{%
  \IfNoValueTF{#1}
    {\m{\Operator}}
    {\m{\Operator \left(#1\right)}}%
}

\newcommand{\Approximator}{\m{A}}
\NewDocumentCommand\Ap{g}{%
  \IfNoValueTF{#1}
    {\m{\Approximator}}
    {\m{\Approximator \left(#1\right)}}%
}

\newcommand{\calo}{\m{{I_o}}}

\newcommand\struclatp[1]{\m{\lattice^{\trans{\voc}}_{#1}}}
	
\NewDocumentCommand\struclat{g}{%
	\IfNoValueTF{#1}
		{\m{\lattice^\voc_\calo}}
		{\m{\lattice^\voc_{#1}}}
}
\NewDocumentCommand\domstruclat{g}{%
	\IfNoValueTF{#1}
		{\m{\lattice^\voc_\calo}}
		{\m{\lattice^\voc_{#1}}}
}

\newcommand{\indomain}{\m{\mathcal{U}}}

\newcommand{\creation}{\m{\mathrm{Create}}}
\newcommand{\selected}{\m{\mathrm{Sel}}}
\newcommand{\choosefirst}{\m{\mathrm{Choose1}}}
\newcommand{\ctinter}[2]{\m{#1_{\mathrm{ct}}^{#2}}}

\newcommand{\ptinter}[2]{\m{#1_{\mathrm{pt}}^{#2}}}

\NewDocumentCommand\proj{g}{%
	\IfNoValueTF{#1}
		{\m{p}}
		{\m{p(#1)}}
}

\newcommand{\domain}{\m{D}}

\newcommand{\selection}{\selected}

\newcommand{\valu}[2]{\m{\mathrm{eff}_{#1,#2}}}
\newcommand{\valich}{\m{\mathrm{eff}_{\I,\ch}}}

\newcommand{\valueof}[1]{\m{\valich(#1)}}
\newcommand{\valueofs}[1]{\m{\mathrm{eff}_{\struct,\ch}(#1)}}
\newcommand{\Valueof}[3]{\m{\valu{#1}{#2}(#3)}}

\newcommand{\Apch}{\m{\Ap_{\ch}}}

\newcommand{\modelsaft}{\m{\models}}

\newacronym{CNF}{CNF}{Causal Normal Form}

\keyword{naming}{naming}
\keyword{Naming}{Naming}
\keyword{Chfun}{$\D$-selection}
\keyword{chfun}{$\D$-selection}
\keyword{Chpfun}{$\relat{\D}$-selection}
\keyword{chpfun}{$\relat{\D}$-selection}
\keyword{welldef}{well-defined}
\newcommand{\ch}{\m{\zeta}}

\newcommand{\chn}[1]{\m{\ch^{new}_{#1}}}
\newcommand{\chs}[1]{\m{\ch^{sel}_{#1}}}

\newcommand{\cho}[1]{\m{\ch^{or}_{#1}}}
\newcommand{\chin}{\m{\ch^{in}}}
\newcommand{\domainof}[1]{\m{D^{#1}}}

\newcommand{\cct}[1]{\m{#1_{ct}}}
\newcommand{\cpt}[1]{\m{#1_{pt}}}

\newcommand{\psetsym}{\m{\mathcal{S}}}

\newcommand{\resset}[2]{\m{r(#1,#2)}}

\usepackage{amsthm}

\theoremstyle{plain}
\newtheorem{thm}{Theorem}[section]

\newtheorem{corollary}[thm]{Corollary}

\newtheorem{lemma}[thm]{Lemma}

\newtheorem{definition}[thm]{Definition}

\theoremstyle{definition}

\newtheorem{example}[thm]{Example}

\newtheorem{ex*}{Example}
\newcommand{\cont}{continued}

\newtheoremstyle{example_contd}
{0.5\topsep} 
{0.5\topsep} 
{\normalfont}
{0pt}
{\bfseries}
{.}
{5pt plus 1pt minus 1pt}
{\thmname{#1} \thmnumber{ #2}\thmnote{#3} (\cont)}

\newtheoremstyle{plain_contd}
{0.5\topsep} 
{0.5\topsep} 
{\itshape}
{0pt}
{\bfseries}
{.}
{5pt plus 1pt minus 1pt}
{\thmname{#1} \thmnumber{ #2}\thmnote{#3} (\cont)}

\theoremstyle{example_contd} 

\theoremstyle{plain_contd} 

\theoremstyle{plain}

\RenewDocumentCommand\CLRule{m+g}{%
		\IfNoValueTF{#2}
			{\m{& #1.}}
			{\m{& #1 \cause #2 }}}

\lstdefinelanguage{br}{
		morekeywords=[1]{if,then,else,for,end,while,when,rule},
	}
	\newcommand{\Catom}[1]{#1}
	\NewDocumentCommand\Call{g+g+g}{%
	\IfNoValueTF{#1}%
	    {\m{\mathbf{All}}}%
	    {\m{\mathbf{All\,}#1[#2]: #3}}%
	 }

	\NewDocumentCommand\Cand{g+g}{%
	\IfNoValueTF{#1}%
	    {\m{\mathbf{And}}}%
	    {\m{#1\mathbf{\,And\,}#2}}%
	}
	\NewDocumentCommand\Csel{g+g+g+g}{%
	\IfNoValueTF{#1}%
	    {\m{\mathbf{Select}}}%
		{\IfNoValueTF{#4}%
			{\m{\mathbf{Select\,}#1[#2]: #3}}%
			{\m{\mathbf{Select}_{#1}\,#2[#3]: #4}}%
	}}
	\NewDocumentCommand\Cor{g+g+g+g+g}{%
	\IfNoValueTF{#1}%
		{\m{\mathbf{Or}}}%
		{\m{#1\mathbf{\,Or\,}#2}%
			\IfValueTF{#3}{\mathbf{\,Or\,}#3%
				\IfValueTF{#4}{\mathbf{\,Or\,}#4%
					\IfValueTF{#5}{\mathbf{\,Or\,}#5}{}%
				}{}%
			}{}%
		}%
	}
	\NewDocumentCommand\Cnew{g+g+g}{%
	  \IfNoValueTF{#1}%
	    {\m{\mathbf{New}}}%
	    {\IfNoValueTF{#3}%
			{\m{\mathbf{New\,}#1: #2}}%
			{\m{\mathbf{New\,}#1[#2]: #3}}}}%

	\NewDocumentCommand\Cif{g+g}{%
	\IfNoValueTF{#1}%
	    {rule}%
	    {\m{#2 \leftarrow #1 }}%
	}

\keyword{CEE}{CEE}

\renewcommand{\CLRule}[2]{\Cif{#2}{#1}}

\renewcommand{\subs}[2]{#2:#1}

\renewcommand{\struct}{\m{J}}

\newcommand{\tcaused}{endogenous\xspace}
\newcommand{\topen}{exogenous\xspace}

\newcommand{\xemph}[1]{\emph{#1}} 

\keyword{cset}{effect set} 
\keyword{poscset}{possible effect set} 

 \setboolean{commentmargin}{false}

\newcommand{\CVC}[2]{\m{#1\langle#2\rangle}}

\Catom

\newcommand{\relat}[1]{\m{rel(#1)}}

\newcommand{\trans}[1]{\m{#1^t_\D}}
\newcommand{\sel}[1]{\m{#1^s_\D}}
\newcommand{\seltheo}{\m{\pred{S}_\D}}

\newcommand{\initpred}{\m{\mathcal{S}}}

\newcommand{\PPP}{\m{\bar{P}}}

\newcommand{\condition}[1]{\m{\pred{Rel}_{#1}}}
\newcommand{\pos}[1]{\m{\pred{Succ}_{#1}}}

\keyword{nestingdepth}{nesting depth}
\keyword{Nestingdepth}{Nesting depth}

 \ouracronym{nestnf}{NestNF}{Nesting Normal Form} 
\ouracronym{deff}{DefF}{Definition Form}

\renewcommand{\foclog}{\logicname{FO(C)}}
\newcommand{\unnest}{\texttt{Unnest}}

\newcommand{\Dt}{\m{\D^t}}

\newcommand\trivialproof[1]{#1}
\newcommand\smallcitation[1]{#1}

\renewcommand{\I}{\m{I}}
\renewcommand{\J}{\m{J}}

\usepackage{aaai}
\usepackage{times}
\usepackage{helvet}
\usepackage{courier}
\title
{Inference in the 
\foclog Modelling Language}
\author{Bart Bogaerts \and Joost Vennekens \and Marc Denecker\\
Department of Computer Science, KU Leuven \\
\{bart.bogaerts, joost.vennekens, marc.denecker\}@cs.kuleuven.be\\
\AND
Jan Van den Bussche\\
Hasselt University \& transnational University of Limburg\\
 jan.vandenbussche@uhasselt.be}

\begin{document}
	\nocopyright
\maketitle
 \bibliographystyle{aaai}

\setboolean{nocomments}{true}

\begin{abstract}
Recently, \foclog, the integration of \clog with classical logic, was introduced as a knowledge representation language. 
Up to this point, no systems exist that perform inference on \foclog, and very little is known about properties of inference in \foclog. 
In this paper, we study both of the above problems. 
We define normal forms for \foclog, one of which corresponds to \foid. 
We define transformations between these normal forms, and show that, using these transformations, several inference tasks for \foclog can be reduced to inference tasks for \foid, for which solvers exist. 
We implemented this transformation and hence, created the first system that performs inference in \foclog. 
We also provide results about the complexity of reasoning in \foclog.

\end{abstract}

\section{Introduction}\label{sec:intro}
Knowledge Representation and Reasoning is a subfield of Artificial Intelligence  concerned with
 two tasks:
defining modelling languages that allow intuitive, clear, representation of knowledge and 
developing inference tools to reason with this knowledge.
Recently,  \clog was introduced with a strong focus on the first of these two goals \mycite{clog}. 
\clog  has an expressive recursive syntax suitable for expressing various forms of non-monotonic reasoning: disjunctive information in the context of closed world assumptions, 
non-deterministic inductive constructions, causal processes, and ramifications. 
\clog allows for example nested occurrences of causal rules. 

It is straightforward to integrate \fo with \clog, offering an expressive modelling language in which causal processes as well as assertional knowledge in the form of axioms and constraints can be naturally expressed. We call this integration \foclog.\footnote{Previously, this language was called \logicname{FO(C-Log)}} 
\foclog fits in the \fodot research project \cite{ASPOCP/Denecker12}, which aims at integrating expressive language constructs with a Tarskian model semantics in a unified language.

An example of a \clog expression is the following
\begin{ltheo}
	\begin{ldef}
		\Call{p}{\pred{Apply}(p) \land \pred{PassedTest}(p)}{ \pred{PermRes}(p).}\\
		\CLRule{(\Csel{p}{\pred{Participate}(p)}{\pred{PermRes}(p)})}{\pred{Lott}.}
	\end{ldef}
\end{ltheo}
This describes that all persons who pass a naturalisation test obtain permanent residence in the U.S., and that one person who participates in the green card lottery also obtains residence.
The person that is selected for the lottery can either be one of the persons that also passed the naturalisation test, or someone else. 
There are local closed world assumptions: in the example, the endogenous predicate \pred{PermRes} only holds for the people passing the test and at most one extra person. 
We could add an \fo constraint to this theory, for example $\forall p:\pred{Participate}(p)\limplies \pred{Apply}(p)$. 
This results in a \foclog theory; a structure is a model of this theory if it is a model of the \clog expression and no-one participates in the lottery without applying the normal way.

So far, very little is known about inference in \foclog. 
No systems exist to reason with \foclog, and complexity of inference in \foclog has not been studied. 
This paper studies both of the above problems.

The rest of this paper is structured as follows: in Section~\ref{sec:prelims}, we repeat some preliminaries, including a very brief overview of the semantics of \foclog.
In Section~\ref{sec:transfo} we define normal forms on \foclog and transformations between these normal forms. We also argue that one of these normal forms corresponds to \foid~\cite{tocl/DeneckerT08} and hence, that \idp~\cite{WarrenBook/DeCatBBD14} can be seen as the first \foclog-solver.
In Section~\ref{sec:ex} we give an example that illustrates both the semantics of \foclog and the transformations.
Afterwards, in Section~\ref{sec:complexity}, we define inference tasks for \foclog and study their complexity. 
We conclude in Section~\ref{sec:concl}.

\todo{ List the contributions concisely in the introduction.}

\section{Preliminaries}\label{sec:prelims}
We assume familiarity with basic concepts of \FO. 
Vocabularies, formulas, and terms are defined as usual. 
A \voc-structure \I interprets all symbols (including variable symbols) in \voc;  \domainof{\I} denotes the domain of \I and
$\sigma^\I$, with $\sigma$ a symbol in $\voc$, 
the interpretation of $\sigma $ in \I.
We use $\I[\subs{v}{\sigma}]$ for the structure $\J$ that equals \I, except on $\sigma$: $\sigma^\J=v$.
\emph{Domain atoms} are atoms of the form $P(\ddd)$ where the $d_i$ are domain elements.
We use restricted quantifications\smallcitation{, see e.g.  \cite{PreyerP02}}. In \FO, these are formulas of the form
$\univBQ{x}{\psi}: \varphi\label{formula:BinQuantForall}$ or $\exisBQ{x}{\psi}: \varphi\label{formula:BinQuantExists}$,
meaning that $\varphi$ holds for all (resp.~for some) $x$ such that $\psi$ holds.
The above expressions are syntactic sugar for
$\forall x: \psi \limplies \varphi$ and $ \exists x: \psi \land \varphi,$
but such a reduction is not possible for other restricted quantifiers in \clog.
We call $\psi$ the \emph{qualification} and $\varphi$ the \emph{assertion} of the restricted quantifications.
From now on, let \voc be a relational vocabulary, i.e., \voc consists only of predicate, constant and variable symbols.

Our logic has a standard, two-valued Tarskian semantics, which means that models represent possible states of affairs. 
Three-valued logic  with partial domains is used as a technical device to express intermediate stages of causal processes. 
A truth-value is one of the following: $\{\ltrue,\lfalse,\lunkn\}$, where $\lfalse^{-1}=\ltrue, \ltrue^{-1}=\lfalse$ and $\lunkn^{-1}=\lunkn$.
Two partial orders are defined on truth values: the precision order \leqp, given by $\lunkn\leqp\ltrue$ and $\lunkn\leqp\lfalse$ and the truth order $\lfalse\leq\lunkn\leq\ltrue$. 
	Let $D$ be a set, a \emph{partial set} \psetsym in $D$ is a function from $D$ to truth values.
We identify a partial set with a tuple $(\cct{\psetsym},\cpt{\psetsym})$ of two sets, where the \emph{certainly true set} 
$\cct{\psetsym}$ is $\{x\mid \psetsym(x)=\ltrue\}$ and the \emph{possibly true set} $\cpt{\psetsym}$ is $\{x\mid \psetsym(x)\neq\lfalse\}$. 
The union, intersection, and subset-relation of partial sets are  defined pointwise. \marc{leg meer uit in finale versie}
For a  truth value $v$, we define the restriction of a partial set $\psetsym$ to this truth-value, denoted $\resset{\psetsym}{v}$, as the partial set mapping every $x\in D$ to $\min_\leq(\psetsym(x),v)$. Every set $S$ is also a partial set, namely the tuple $(S,S)$.

A \emph{partial \voc-structure} \I consists of 
1) a \emph{domain} \domainof{\I}: a partial set of elements, and
2) a mapping associating a value to each symbol in \voc; for constants and variables, this value is in \ctinter{\domain}{\I}, 
for predicate symbols of arity $n$, this is a partial set $\inter{P}{\I}$ in $( \ptinter{\domain}{\I})^n$.
We often abuse notation and use the domain \domain as if it were a predicate. 
A partial structure \I is \emph{two-valued} if for all predicates $P$ (including \domain), $\ctinter{P}{\I}=\ptinter{P}{\I}$. 
There is a one-to-one correspondence between two-valued partial structures and structures.
If \I and \J are two partial structures with the same interpretation for constants, we call \I more precise than \J (\I\geqp \J) if for all its predicates $P$ (including \domain),
$\ctinter{P}{\I}\supseteq\ctinter{P}{\J}$ and $\ptinter{P}{\I}\subseteq\ptinter{P}{\J}$.

\bart{Shorter: removed: Intuitively, the value of a quantification $\forall x: \varphi$, is the value of the restricted quantification
$\forall x[\domain(x)]: \varphi$, where \domain refers to the (three-valued) domain. }

\begin{definition}\label{def:kleeneval}
We define the value of an \fo formula $\varphi$ in a partial structure $\I$ inductively based on the Kleene truth tables\smallcitation{ \cite{Kleene38}}. 
\begin{itemize}
	\item $P(\ttt)^\I = P^{\I}(\ttt^\I)$,
	\item $(\neg\varphi)^\I = ((\varphi)^\I)^{-1}$
	\item $(\varphi\land\psi)^\I = \min_\leq\left(\varphi^\I,\psi^\I\right)$
	\item $(\varphi\lor\psi)^\I = \max_\leq\left(\varphi^\I,\psi^\I\right)$
	\item $(\forall x:\varphi)^\I = \min_\leq\bigl\{\max(\domainof{\I}(d)^{-1},\varphi^{\I[\subs{d}{x}]})\mid d\in \ptinter{\domain}{\I}\bigr\}$
	\item $(\exists x:\varphi)^\I = \max_\leq\bigl\{\min(\domainof{\I}(d),\varphi^{\I[\subs{d}{x}]}) \mid d\in \ptinter{\domain}{\I}\bigr\}$
\end{itemize}
\end{definition}

In what follows we briefly repeat the syntax and formal semantics of \clog. For more details, an extensive overview of the informal semantics of \CEEs, and examples of \CEEs, we refer to \mycite{clog}.

\subsection{Syntax of \clog}

\begin{definition}
 \emph{Causal effect expressions} (\CEE) are defined inductively as follows:
\begin{itemize}
	\item if $P(\ttt)$ is an atom, then $\Catom{P(\ttt)}$ is a \CEE,
	\item if $\varphi$ is an \fo formula and $C'$ is a \CEE, then $\Cif{\varphi}{C'}$ is a \CEE,
	\item if $C_1$ and $C_2$ are \CEEs, then $\Cand{C_1}{C_2}$ is a \CEE,
	\item if $C_1$ and $C_2$ are \CEEs, then $\Cor{C_1}{C_2}$  is a \CEE,
	\item if $x$ is a variable, $\varphi$ is a first-order formula and $C'$ is a \CEE, then $\Call{x}{\varphi}{C'}$  is a \CEE,
	\item if $x$ is a variable, $\varphi$ is a first-order formula and $C'$ is a \CEE, then $\Csel{x}{\varphi}{C'}$ is a \CEE,
	\item if $x$ is a variable and $C'$ is a \CEE, then $\Cnew{x}{C'}$ is a \CEE.
\end{itemize}
\end{definition}

We call a \CEE an \emph{atom-}  (respectively \emph{\Cif-},   \emph{\Cand-}, \mbox{\emph{\Cor-},}   
\emph{\Call-}, \emph{\Csel-} or \emph{\Cnew-expression}) if it is of the corresponding form. 
We call a predicate symbol $P$ \xemph{\tcaused}in $C$ if $P$ occurs as the symbol of a (possibly nested) atom-expression in $C$.
All other symbols are called \xemph{\topen}in $C$.  
An occurrence of a variable $x$ is \emph{bound} in a \CEE if it occurs in the scope of a quantification over that variable 
($\forall x$, $\exists x$, $\Call\, x$, $\Csel\, x$, or $\Cnew\, x$) and \emph{free} otherwise. A variable is \emph{free} in a \CEE if it has free occurrences. 
A \emph{causal theory}, or \emph{\foidplus theory} is a \CEE without free variables. 
By abuse of notation, we often represent a causal theory as a finite set of \CEEs; the intended causal theory is the \Cand-conjunction of these \CEEs.
We often use \D for a causal theory and $C$, $C'$, $C_1$ and $C_2$ for its subexpressions. 
We stress that the connectives in \CEEs differ from their \fo counterparts. 
E.g., in the example in the introduction, the \CEE expresses that there is a cause for several persons to become American (those who pass the test and maybe one extra lucky person). 
This implicitly also says that every person without cause for becoming American is not American. As such \clog-expressions are highly non-monotonic. \marc{kan verwarrend zijn, leg meer uit in journal versie}

\subsection{Semantics of \clog}

\begin{definition}
Let \D be a causal theory;
we associate a  parse-tree with \D.
An \emph{occurrence} of a \CEE $C$ in \D is a node in the
parse tree of \D labelled with $C$.  The \emph{variable context} of an occurrence of a \CEE $C$ in  \D 
is the sequence of quantified variables as they occur on
the path from \D to $C$ in the parse-tree of \D. If $\xxx$ is the variable context of $C$ in $\D$, we denote $C$ as $\CVC{C}{\xxx}$ and the length of $\xxx$ as $n_C$. 
\end{definition}
For example, the variable context of $P(x)$ in $
        \Csel{y}{Q(y)}{\Call{x}{Q(x)}{P(x)}}$
is $[y,x]$. 
{\em Instances} of an occurrence $\CVC{C}{\xxx}$ correspond to assignments $\ddd$ of domain elements to $\xxx$. 

\begin{definition}\label{def:chfun}
Let \D be a causal theory and \domain a set. 
A \xemph{\chfun}\ch in \domain consists of 
\begin{itemize}
 \item for every occurrence $C$ of a \Csel-expression in \D, a total function $\chs{C}:\domain^{n_C}\to \domain$,
 \item for every occurrence $C$ of a \Cor-expression  in \D, a total function $\cho{C}:\domain^{n_C}\to \{1,2\}$,
\item for every occurrence $C$ of a \Cnew-expression in \D, an injective partial function $\chn{C}:\domain^{n_C}\to \domain$.
\end{itemize}
such that furthermore the images of all functions \chn{C} are disjoint (i.e., such that every domain element can be created only once).

The \emph{initial elements} of \ch are those that do not occur as image of one of the \chn{C}-functions: $\chin=D\setminus \cup_C \text{image}(\chn{C})$, where the union ranges over all occurrences of \Cnew-expressions.\end{definition}

The \cset of a \CEE in a partial structure is a partial set:
it contains information on everything that is caused and  everything that might be caused. 
For defining the semantics a new, unary predicate \indomain is used.

\begin{definition}

Let \D be a \CEE and \struct a partial structure. Suppose \ch is a \chfun in a set $D\supseteq\ptinter{\domain}{\struct}$.
Let $C$ be an occurrence of a \CEE in \D.
The \emph{\cset of $C$} with respect to \struct and \ch is a partial set of domain atoms, defined recursively: 
\begin{itemize}
	\item If $C$ is $P(\ttt)$, then $\valueofs{C} =\{P(\ttt^\struct)\}$,  
	\item if $C$ is  $\Cand{C_1}{C_2}$, then 
$\valueofs{C}=\valueofs{C_1}\cup\valueofs{C_2}$,
	\item if $C$ is $\Cif{\varphi}{C'}$, then 
$\valueofs{C} = \resset{\valueofs{C'}}{\varphi^\struct}$,
	\item if $C$ is $\Call{x}{\varphi}{C'}$, then 
		\begin{itemize}\item[]
$\valueofs{C} =$ $\bigcup\Bigl\{ r\bigl(\valu{\struct'}{\ch}(C'),\min_\leq(\domain^\struct(d),\varphi^{\struct'})\bigr)    
	\mid  d\in \ptinter{\domain}{\struct}$ and $ \struct'=\struct[\subs{d}{x}]\Bigr\}$     
\end{itemize}
	\item if $\CVC{C}{\yyy}$ is  $\Cor{C_1}{C_2}$, then
	\begin{itemize}
	 \item $\valueofs{C} =\valueofs{C_1}$ if $\cho{C}(\yyy^\struct)=1$,
	 \item and $\valueofs{C} =\valueofs{C_2}$ otherwise
	\end{itemize} 
	\item if $\CVC{C}{\yyy}$ is  $\Csel{x}{\varphi}{C'}$,  let $e=\chs{C}(\yyy^\struct)$,  $\struct'=\struct[\subs{e}{x}]$ and $v=\min_\leq(\domain^\struct(e),$ $\varphi^{\struct'})$. Then
$\valueofs{C} =\resset{\valu{\struct}{\ch}(C')}{v}$, 
	\item if $\CVC{C}{\yyy}$ is $\Cnew{x}{C'}$,  then
		\begin{itemize}
		\item $\valueofs{C} =\emptyset$  if $\chn{C}(\yyy^\struct)$ does not denote,
		\item and $\valueofs{C} =\{\indomain(\chn{C}(\yyy^\struct))\}\cup  \valu{\struct'}{\ch}(C')$, where  $\struct'=\struct[\subs{\chn{C}(\yyy^\struct)}{x}]$ otherwise,
		\end{itemize}
\end{itemize}
An instance of an occurrence of a \CEE in \D is \emph{relevant} if it is encountered in the evaluation of \valueof{\D}. 
We say that $C$ \emph{succeeds}\footnote{Previously, we did not say that $C$ ``succeeds'', but that the effect set ``is a possible effect set''. We believe this new terminology is more clear.}  with  $\ch$ in \struct if for all relevant occurrences \CVC{C}{\yyy} of  \Csel-expressions, $\chs{C}(\yyy^\struct)$ satisfies the qualification of $C$ and for all relevant instances \CVC{C}{\yyy} of \Cnew-expressions, $\chn{C}(\yyy^\struct)$ denotes. \marc{leg meer uit in journal versie}
\end{definition}

Given a structure \I (and a \chfun \ch), two lattices are defined:  $\struclat{\I,\zeta}$   denotes the set of all \voc-structures \struct with $\chin\subseteq\domainof{\struct}\subseteq\domainof{\I}$ 
such that for all \topen symbols $\sigma$ of arity $n$:
$\sigma^\struct=\sigma^\I\cap {(\domainof{\struct})}^n$.
This set is equipped with the truth order.
And \struclat{\I} denotes the sublattice of $\struclat{\I,\zeta}$ consisting of all structures in $\struclat{\I,\zeta}$ with domain equal to \domainof{\I}.
\marc{In journal versie: uitleggen dat dit wel degelijk een lattice is: oorsneede van structuren terug structuur}

A partial structure corresponds to an element of the bilattice  $(\struclat{\I,\zeta})^2$; the bilattice is equipped with the precision order.

\begin{definition}
Let  \I be  a structure and \ch a \chfun  in \domainof{\I}.
 The \emph{partial immediate causality operator} $\Apch$ is the operator on  $(\struclat{\I,\zeta})^2$ that  sends partial structure $\struct$  to a partial structure $\struct'$   such that
\begin{itemize}
	\item $\domain^{\struct'}(d) = \ltrue $ if $d\in\chin$ and $\domain^{\struct'}(d) =\Valueof{\struct}{\ch}{\D}(\indomain(d))$ otherwise
	\item for \tcaused symbols $P$, $\inter{P(\ddd)}{\struct'} = \Valueof{\struct}{\ch}{\D}(P(\ddd))$.
\end{itemize}
\end{definition}

Such operators have been studied intensively in the field of Approximation Fixpoint Theory \cite{DeneckerBV12}; and for such operators, the well-founded fixpoint has been defined in \cite{DeneckerBV12}. 
The semantics of \clog is defined in terms of this well-founded fixpoint in \mycite{clog}:

\begin{definition}
Let \D be a causal theory. 
We say that structure $\I$ is a \emph{model} of \D (notation $\I\modelsaft\D$) if there exists a \chfun \ch such that 	 
(\I,\I) is the well-founded fixpoint of $\Apch$, and $\D$ succeeds with \ch in \I.
\end{definition}

 \foclog is the integration of \fo and \clog. An \foclog theory consists of a set of causal theories and \fo sentences. A structure \I is a model of an \foclog theory if it is a model of all its causal theories and \fo sentences.
In this paper, we  assume, without loss of generality, that an \foclog theory \theory has exactly one causal theory. 

 \section{A Transformation to DefF}\label{sec:transfo}
In this section we present normal forms for \foclog and transformations between these normal forms. 
 The transformations we propose preserve equivalence modulo newly introduced predicates:
\begin{definition}
	Suppose $\voc\subseteq\voc'$ are vocabularies,
	$\theory$ is an \foclog theory over $\voc$ and $\theory'$ is an \foclog theory over $\voc'$.
	We call $\theory$ and $\theory'$ \emph{$\voc$-equivalent} if each model of $\theory$, can be extended to a model of $\theory'$ and the restriction of each model of $\theory'$ to $\voc$ is a model of $\theory$. 
\end{definition}

From now on, we use $\Call{\xxx}{\varphi}{C'}$, where $\xxx$ is a tuple of variables as syntactic sugar for \Call{x_1}{\ltrue}{\Call{x_2}{\ltrue}{\dots \Call{x_n}{\varphi}{C'}}}, and similar for $\Csel$-expressions. 
If $\xxx$ is a tuple of length $0$, $\Call{\xxx}{\varphi}{C'}$ is an abbreviation for $\Cif{\varphi}{C'}$. 
It follows directly from the definitions that \Cand and \Cor are associative, hence we use 
$\Cand{C_1}{\Cand{ C_2 }{ C_3}}$ as an abbreviation for $\Cand{(C_1}{\Cand{ C_2) }{ C_3}}$ and for $\Cand{C_1}{\Cand{( C_2 }{ C_3)}}$, and similar for $\Cor$-expressions.
\marc{maak hier maar wat meer woorden aan vuil in de journal versie: leg eerst uit dat all x phi C hetzelfde is als all x true C <- phi en dan pas naar die ariteit nul gaan etcetera}

%
%

\subsection{Normal Forms}

\begin{definition}
Let $C$ be an occurrence of a \CEE in $C'$. The \emph{nesting depth} of $C$ in $C'$ is the depth of $C$ in the parse-tree of $C'$.
In particular, the nesting depth of $C'$ in 
$C'$ is always $0$.
The \emph{height} of $C'$ is the maximal nesting depth of occurrences of \CEEs in $C'$. In particular, the height of atom-expressions is always $0$.
\end{definition}

\begin{example}
Let $\D$ be 
$\Cand{A}{(\Cor{(\Call{x}{P(x)}{Q(x)})}{B})}$.
The nesting depth of $B$ in $\D$ is $2$ and the height of \D is $3$.
\end{example}

\begin{definition}
A \clog theory is \emph{creation-free} if it does not contain any \Cnew-expressions, it is \emph{deterministic} if it is creation-free and it does not contain any \Csel or \Cor-expressions. An \foclog is \emph{creation-free} (resp. \emph{deterministic}) if its (unique) \clog theory is.
\end{definition}

\begin{definition}\label{def:normalforms}
A \clog theory is in \xemph{\nestnf}if it is of the form 
$\Cand{C_1}{\Cand{C_2}{\Cand{C_3}{\dots}}}$
where each of the $C_i$ is of the form
$\Call{\xxx}{\varphi_i}{C_i'}$
and each of the $C_i'$ has height at most one. 
A \clog theory \D is in \xemph{\deff}if it is in \nestnf and each of the $C_i'$ have height zero, i.e., they are atom-expressions.
An \foclog theory is \nestnf (respectively \deff) if its corresponding \clog theory is. 
\end{definition}


\begin{thm}\label{thm:main}
	Every \foclog theory over \voc is \voc-equivalent with an \foclog theory in \deff. 
\end{thm}

We will prove this result in 3 parts: in Section \ref{ssec:clog2nestnf}, we show that every \foclog theory can be transformed to  \nestnf, in Section     \ref{ssec:nestnf2deterministic}, we show that every theory in \nestnf can be transformed into a deterministic theory and in Section \ref{ssec:deterministic2deff}, we show that every deterministic theory can be transformed to \deff. 
The \fo sentences in an \foclog theory do not matter for the normal forms, hence most results focus on the \clog part of \foclog theories.

\subsection{From Deterministic \foclog to DefF}\label{ssec:deterministic2deff}
\begin{lemma}\label{lem:pushCallCand}
	Let \D be a \clog theory. Suppose $C$ is an occurrence of an expression $\Call{\xxx}{\varphi}{\Cand{C_1}{C_2}}$. Let $\D'$ be the causal theory obtained from $\D$ by replacing $C$ with \Cand{(\Call{\xxx}{\varphi}{C_1})}{(\Call{\xxx}{\varphi}{C_2})}.
	Then \D and $\D'$ are equivalent.
\end{lemma}
\trivialproof{
\begin{proof}
	It is clear that \D and $\D'$ have the same selection functions. Furthermore, it follows directly from the definitions that given such a selection, the defined operators are equal.
\end{proof} }

Repeated applications of the above lemma yield:
\begin{lemma}\label{lem:deterministic2deff}
	Every deterministic \foclog theory is equivalent with an \foclog theory in \deff.
\end{lemma}

\subsection{From NestNF to Deterministic \foclog}\label{ssec:nestnf2deterministic}
%
\begin{lemma}\label{lem:creation-free2deterministic}
		If \theory is an \foclog theory in  \nestnf over \voc, then \theory is \voc-equivalent with a deterministic \foclog theory.
\end{lemma}

We will prove Lemma \ref{lem:creation-free2deterministic} using a strategy that replaces a \chfun by an interpretation of new predicates (one per occurrence of a non-deterministic \CEE).
The most important obstacle for this transformation are  \Cnew-expressions. 
In deterministic \clog, no constructs influence the domain.
This has as a consequence that the immediate causality operator for a deterministic \clog theory is defined in a lattice of structures with fixed domain, while in general, the operator is defined in a lattice with variable domains. 
In order to bridge this gap, we use two predicates to describe the domain, \initpred are the initial elements and \indomain are the created, the  union of the two is the domain.
Suppose a \clog theory \D over vocabulary \voc is given. 

\begin{definition}
 We define the \emph{\D-selection vocabulary} \sel{\voc} as the vocabulary consisting of:
  \begin{itemize}
   \item a unary predicate \initpred,
   \item for every occurrence $C$ of a \Cor-expression in \D, a new $n_C$-ary predicate $\choosefirst_C$,
   \item for every occurrence $C$ of a \Csel-expression in \D, a new $(n_C+1)$-ary predicate $\selected_C$,
   \item for every occurrence $C$ of a \Cnew-expression in \D, a new $(n_C+1)$-ary predicate $\creation_C$,
  \end{itemize}
\end{definition}

Intuitively, a \sel{\voc}-structure corresponds to a \chfun: \initpred correspond to \chin, $\choosefirst_C$ to $\cho{C}$, $\selected_C$ to \chs{C} and $\creation_C$ to \chn{C}.
  
   \begin{lemma}\label{lem:chfun-seltheo}
There exists an \fo theory  \seltheo over  \sel{\voc} such that there is a one-to-one correspondence between \chfuns in \domain and models of \seltheo with domain \domain. 
  \end{lemma} 
  \begin{proof}
This theory contains sentences that express that $\selected_C$ is functional, and that $\creation_C$ is a partial function. It is straightforward to do this in \fo (with among others, constraints such as $\forall \xxx : \exists y: \selected_C(\xxx,y)$).
Furthermore, it is also easy to express that the $\creation_C$ functions are injective, and that different \Cnew-expressions create different elements. 
Finally, this theory relates $\initpred$ to the $\creation_C$ expressions: $\forall y: \initpred(y) \lequiv \lnot \bigvee_C (\exists \xxx : \creation_{C}(\xxx,y))$ where the disjunction ranges over all occurrences $C$ of \Cnew-expressions.  
  \end{proof}

The condition that a causal theory succeeds can also be expressed as an FO theory.
For that, we need one more definition.
\begin{definition}
	Let \D be a causal theory in \nestnf and let $C$ be one of the $C_i'$ in definition \ref{def:normalforms}, then we call $\varphi_i$ (again, from definition \ref{def:normalforms}) the \emph{relevance condition} of $C$ and denote it \condition{C}.
\end{definition}
%

In what follows, we define one more extended vocabulary.
First, we use it to express the constraints that \D succeeds and afterwards, for the actual transformation.

\begin{definition}
	 The \emph{\D-transformed vocabulary} \trans{\voc} is the disjoint union of \voc and \sel{\voc} extended with the unary predicate symbol \indomain.
\end{definition}

%

\begin{lemma}\label{lem:posseffectsettheory}
	Suppose \D is a causal theory in \nestnf, and \ch is a \chfun with corresponding \sel{\voc}-structure $M$. There exists an \fo theory \pos{\D} such that for every (two-valued) structure $\I$ with $\I|_{\sel{\voc}}=M$, \D succeeds with respect to \I and \ch iff  $\I\models \pos{\D}$.
\end{lemma}
\begin{proof}
\D is in \nestnf; for every of the $C'_i$ (as in Definition \ref{def:normalforms}), $\condition{C'_i}$ is true in \I if and only if $C'_i$ is relevant.  \marc{let op met vrije variabelen}\bart{geen probleem want \I interpreteert ze, leg uit in journal versie}
Hence, for \pos{\D} we can take the \FO theory consisting of the following sentences:
\begin{itemize}
	\item $\forall \xxx: \condition{C}\limplies \exists y: \creation_{C}(\xxx,y)$, for all \Cnew-expressions \CVC{C}{\xxx} in \D,
	\item $\forall \xxx: \condition{C}\limplies \exists y: (\selected_{C}(\xxx,y)\land \psi)$, for all \Csel-expressions \CVC{C}{\xxx} of the form \Csel{y}{\psi}{C'} in $\D$.\qedhere
\end{itemize}
\end{proof}

Now we describe the actual transformation: we translate every quantification into a relativised version, make explicit that a \Cnew-expression causes an atom $\indomain(d)$, and eliminate all non-determinism using the predicates in \sel{\voc}.

\begin{definition}
 Let $\Delta$ be a \clog theory over \voc in \nestnf. 
 The \emph{transformed theory} \Dt is the theory obtained from \D by applying the following transformation:
\begin{itemize}
\item first replacing all quantifications  $\alpha x [\psi]:\chi$, where $\alpha \in \{\forall, \exists, \Csel, \Call\}$  by $\alpha x[(\indomain(x)\vee\initpred(x))\land\psi]:\chi$ 
\item subsequently replacing each occurrence \CVC{C}{\xxx} of an expression \Cnew{y}{C'} by $\Call{y}{\creation_C(\xxx,y)}{\Cand{\indomain(y)}{C'}}$, 
\item replacing every occurrence $\CVC{C}{\xxx}$ of an expression $\Cor{C_1}{ C_2}$ by 
$(\Cif{\choosefirst_C(\xxx)}{C_1})\Cand (\Cif{\lnot \choosefirst_C(\xxx)}{C_2}),$
	\item and replacing every occurrence $\CVC{C}{\xxx}$ of an expression $\Csel{y}{\varphi}{C'}$ by 
	$\Call{y}{\varphi\land \selection_C(\xxx,y)}{C'}.$
\end{itemize}
\end{definition}

\newcommand{\mch}{\m{m_{\ch}}}

Given a structure \I and a \chfun \ch, there is an obvious lattice morphism  $\mch:\struclat{\I,\ch} \to \struclatp{\I}$ mapping a structure $J$ to the structure $J'$ with domain $\domainof{J'}=\domainof{\I}$ interpreting all symbols in \sel{\voc} according to \ch (as in Lemma \ref{lem:chfun-seltheo}), all symbols in \voc (except for the domain) the same as \I and interpreting \indomain as $\domainof{J}\setminus\inter{\initpred}{J'}$.
\mch can straightforwardly be extended to a bilattice morphism.
\marc{in journal versie: let uit wat morphism, isomorphism,... is en geef eigenschappen}

%
%
%

\begin{lemma}\label{lem:new:operatorsagree}
	Let \ch be a \chfun for \D and \Apch and \Ap be the partial immediate causality operators of \D and \Dt respectively. 
	Let \struct be any partial structure in $(\struclat{\I,\ch})^2$. Then $\mch(\Apch(\struct))=\Ap(\mch(\struct))$.
\end{lemma}
\begin{proof}[Idea of the proof]
\Cnew-expressions $\Cnew{y}{C'}$ in \D have been replaced by \Call expressions causing two subexpressions: $\indomain(y)$ and the $C'$ for exactly the $y$'s that are created according to \ch. 
Furthermore, the relativisation of all other quantifications guarantees that we correctly evaluate all quantifications with respect to the domain of \struct, encoded in  $\initpred\cup\indomain$.

Furthermore, all non-deterministic expressions have been changed into \Call-expressions that are conditionalised by the \chfun; this does not change the  effect set; thus, the operators correspond.
\end{proof}

\begin{lemma}\label{lem:WFModelAgree}
	Let \ch, \Apch and \Ap be as in lemma \ref{lem:new:operatorsagree}. If \I is the well-founded model of \Apch, $\mch({ \I })$ is the well-founded model of \Ap.
\end{lemma}
\begin{proof}
	Follows directly from lemma \ref{lem:new:operatorsagree}: the mapping $\struct\mapsto\mch(\struct)$ is an isomorphism between  \struclat{\I,\ch} and the sublattice of  \struclatp{\I,\ch'} consisting of those structures such that the interpretations of $\initpred$ and \indomain have an empty intersection. 
	As this isomorphism maps \Apch to \Ap, their well-founded models must agree.
\end{proof}

\begin{lemma}\label{lem:nsntf_to_several_theories}
	Let \D be a causal theory in \nestnf, \ch a \chfun for \D and \I a \voc-structure. Then $\I\models\D$ if and only if $\mch({\I})\models\Dt$ and $\mch({\I})\models \seltheo$ and $\mch(\I)\models\pos{\D}$.
\end{lemma}
\begin{proof}
	Follows directly from Lemmas \ref{lem:WFModelAgree}, \ref{lem:chfun-seltheo} and \ref{lem:posseffectsettheory}. 
\end{proof}

\begin{proof}[Proof of Lemma \ref{lem:creation-free2deterministic}]
	Let \D be the \clog theory in \theory.
	We can now take as deterministic theory the theory consisting of \Dt, all \fo sentences in \theory, and the sentence $\seltheo \land \pos{\D}\land \forall x: \initpred(x)\lequiv\lnot\indomain(x)$, where the last formula excludes all structures not of the form $\mch(\I)$ for some \I (the created elements \indomain and the initial elements \initpred should form a partition of the domain). 
\end{proof}

\subsection{From General \foclog to NestNF}\label{ssec:clog2nestnf}

In the following definition we use $\D[C'/C]$ for the causal theory obtained from \D by replacing the occurrence of a \CEE $C$ by $C'$.
\begin{definition}
	Suppose \CVC{C}{\xxx} is an occurrence of a \CEE in \D. 
	With \emph{$\unnest(\D,C)$} we denote the causal theory $\Cand{\D[P(\xxx)/ C]}{\Call{\xxx}{P(\xxx)}{C}}$
	where $P$ is a new predicate symbol.
\end{definition}

\begin{lemma}\label{lem:clog2nestnf} 
	Every \foclog theory is \voc-equivalent with an \foclog theory in \nestnf.
\end{lemma}

\begin{proof}
First, we claim that for every \clog theory over \voc, 	$\D$ and $\unnest(\D,C)$ are \voc-equivalent.
It is easy to see that the two theories have the same \chfuns. 
Furthermore, the operator for $\unnest(\D,C)$ is a part-to-whole monotone fixpoint extension\footnote{Intuitively, a part-to-whole fixpoint extension means that all predicates only depend positively on the newly introduced predicates} (as defined in \cite{VennekensMWD07a}) of the operator for $\D$. In \cite{VennekensMWD07a} it is shown that in this case, their well-founded models agree, which proves our claim.  
The lemma now follows by repeated applications of the claim.\end{proof}

\begin{proof}[Proof of Theorem \ref{thm:main}]
	Follows directly by combining lemmas \ref{lem:clog2nestnf},  \ref{lem:creation-free2deterministic} and \ref{lem:deterministic2deff}.
	For transformations only defined on \clog theories, the extra \fo part remains unchanged.
\end{proof}

\subsection{\foclog and \foid}\label{ssec:foid}
An inductive definition (ID) \cite{tocl/DeneckerT08} is a set of rules of the form $\forall \xxx: P(\ttt)\lrule \varphi$, an \foid theory is a set of \fo sentences and IDs, and an \esoid theory is a theory of the form $\exists \PPP: \theory$, where $\theory$ is an \foid theory. 
A causal theory in \deff corresponds exactly to an ID: the \CEE $\Call{\xxx}{\varphi}{P(\ttt)}$ corresponds to the above rule  and the \Cand-conjunction of such \CEEs to the set of  corresponding rules. 
The partial immediate consequence operator for IDs defined in \cite{tocl/DeneckerT08} is exactly the partial immediate causality operator for the corresponding \clog theory. 
Combining this with Theorem \ref{thm:main}, we find (with \PPP the introduced symbols):
\begin{thm}\label{thm:foid}
	Every \foclog theory is equivalent with an \esoid formula of the form $\exists \PPP: \{\D,\theory\}$, where $\D$ is an ID and \theory is an \FO sentence.
\end{thm}
Theorem \ref{thm:foid} implies that we can use reasoning engines for \foid in order to reason with \foclog, as long as we are careful with the newly introduced predicates.
We implemented a prototype of this transformation in the \idp system \cite{WarrenBook/DeCatBBD14}, it can be found at \cite{url:IDP-CLog}.

 \section{Example: Natural Numbers}\label{sec:ex}
\begin{example}
	Let \voc be a vocabulary consisting of predicates $\pred{Nat}/1, \pred{Succ}/2$ and $\pred{Zero}/1$ 
	and suppose \theory is the following theory: \marc{in journal versie: voeg extra constraint $\forall x: Nat(x)$ toe!}
\begin{ltheo}
	\begin{ldef}
		\Cnew{x}{ \Cand{\pred{Nat}(x)}{\pred{Zero}(x)}}\\
		\Call{x}{\pred{Nat}(x)}{ \Cnew{y}{ \Cand{\pred{Nat}(y)}{\pred{Succ}(x,y)}}}
	\end{ldef}
\end{ltheo}This theory defines a process creating the natural numbers. 
Transforming it to \nestnf yields:
\begin{ltheo}
	\begin{ldef}
		\Cnew{x}{ T_1(x) }\\
		\Call{x}{T_1(x)}{\pred{Nat}(x)}\\
		\Call{x}{T_1(x)}{\pred{Zero}(x)}\\
		\Call{x}{\pred{Nat}(x)}{ \Cnew{y}{ T_2(x,y)}}\\
		\Call{x,y}{T_2(x,y)}{\pred{Nat}(y)}\\
		\Call{x,y}{T_2(x,y)}{\pred{Succ}(x,y)},
	\end{ldef}
\end{ltheo}
where $T_1$ and $T_2$ are auxiliary symbols.
	Transforming the resulting theory into deterministic \clog requires the addition of more auxiliary symbols $\initpred/1, \indomain/1, \creation_1/1$ and $\creation_2/2$ and results in the following \clog theory (together with a set of \FO-constraints):
	
\begin{ltheo}
	\begin{ldef}
		\Call{x}{\creation_1(x)}{\Cand{\indomain(x)}{ T_1(x)} }\\
		\Call{x}{(\indomain(x)\lor \initpred(x))\land T_1(x)}{\pred{Nat}(x)}\\
 		\Call{x}{(\indomain(x)\lor \initpred(x))\land T_1(x)}{\pred{Zero}(x)}\\
 		\Call{x, y}{(\indomain(x)\lor \initpred(x))\land \pred{Nat}(x)\land \creation_2(x,y)}{ \\ \quad \Cand{\indomain(y)}{ T_2(x,y)}}\\
 		\Call{x,y}{(\indomain(x)\lor \initpred(x))\land (\indomain(y)\lor \initpred(y))\land T_2(x,y)}{ \\ \quad \pred{Nat}(y)}\\
 		\Call{x,y}{(\indomain(x)\lor \initpred(x))\land (\indomain(y)\lor \initpred(y))\land T_2(x,y)}{ \\ \quad \pred{Succ}(x,y)}
	\end{ldef}
\end{ltheo}
This example shows that the proposed transformation is in fact too complex.
	E.g., here, almost all occurrences of $\indomain(x)\lor \initpred(x)$ are not needed.
	This kind of redundancies can be eliminated by executing the three transformations (from Sections \ref{ssec:deterministic2deff}, \ref{ssec:nestnf2deterministic} and \ref{ssec:clog2nestnf}) simultaneously. In that case, we would get the simpler deterministic theory:
\begin{ltheo}
	\begin{ldef}
		\Call{ x}{ \creation_1(x)}{ 
			\Cand{ \pred{Nat}(x)}{ \Cand{\pred{Zero}(x)}{\indomain(x) } 			}}\\
		\Call{ x, y}{ (\indomain(x) \lor \initpred(x)) \land \pred{Nat}(x) \land \creation_2(x,y) }{\\ \quad
			\Cand{\pred{Nat}(y)}{\Cand{\pred{Succ}(x,y)}{ \indomain(y)}}	}
	\end{ldef}
	\end{ltheo}
	with several \fo sentences:
	\begin{ltheo}
	\forall x: \indomain(x)\lequiv\lnot\initpred(x)\\
	\forall y: \initpred(y) \lequiv \lnot (\creation_1(y) \vee \exists x: \creation_2(x,y)).\\
	\exists x: \creation_1(x).\\
 	\forall x, y: \creation_1(x)\land \creation_1(y)\limplies x = y.\\
 	\forall x, y, z: \creation_2(x,y)\land \creation_1(x,z)\limplies y = z.\\
 	\forall x, y, z: \creation_1(y)\land \creation_1(x,z)\limplies y = z.\\
 	\forall x[\pred{Nat}(x)]:\exists y: \creation_2(x,y).
\end{ltheo}  
These sentences express the well-known constraints on \nat: there is at least one natural number (identified by $\creation_1$), and every number has a successor. Furthermore the initial element and the successor elements are unique, and all are different. Natural numbers are defined as zero and all elements reachable from zero by the successor relation. 
	The theory we started from is much more compact and much more readable than any \foid theory defining natural numbers. This shows the Knowledge Representation power of \foidplus.
\end{example}
	
 \section{Complexity Results}\label{sec:complexity}					
In this section, we provide complexity results. 
We focus on the \clog fragment of \foclog here, since complexity for \fo is  well-studied. 
First, we formally define the inference methods of interest.
\subsection{Inference Tasks}

\begin{definition}
	The \emph{model checking} inference takes as input a \clog theory \D and a finite (two-valued) structure \I.
	It returns true if $\I\models\D$ and false otherwise.
\end{definition}
\todo{1 reviewer vond mx niet duidelijk: dat dit het uiteindelijk resultaat ook eindig is...}
\begin{definition}
	The \emph{model expansion} inference takes as input a \clog theory \D and a partial structure \I with finite two-valued domain. 
	It returns a model of \D more precise than \I if one exists and  ``unsat'' otherwise.
\end{definition}


\begin{definition}
	The \emph{endogenous model expansion} inference is a special case of model expansion where \I is two-valued on \topen symbols of \D and completely unknown on \tcaused symbols.
\end{definition}

The next inference is related to database applications. 
In the database world, languages with object creation have also been defined \cite{aw/AbiteboulHV95}. 
A query in such a language can create extra objects, but the interpretation of exogenous symbols (tables in the database) is fixed, i.e., exogenous symbols are always false on newly created elements.

\begin{definition}
	The \emph{unbounded query} inference takes as input  a  \clog theory \D, a partial structure \I with finite two-valued domain such that \I is two-valued on \topen symbols of \D and completely unknown on \tcaused symbols of \D, and a propositional atom $P$. This inference returns true if there exist i) a structure \J, with $\domainof{\J}\supseteq\domainof{\I}$, $\sigma^\J=\sigma^\I$ for exogenous symbols $\sigma$,  and $P^\J=\ltrue$  and ii) a \chfun \ch in \domainof{\J} with $\chin=\domainof{\I}$, such that \J is a model of \D with \chfun \ch. It returns false otherwise.
\end{definition}

\subsection{Complexity of Inference Tasks}
In this section, we study the datacomplexity of the above inference tasks, i.e., the complexity for fixed $\D$. 

\todo{Given the equivalence of FO(C) and FO(ID) I guess the complexity results holds
also for FO(ID). If so, please state it in the paper or connect these with
previous results for FO(ID).}

\begin{lemma}\label{lem:ApplyApch}
	For a finite structure \I, computing $\Apch(\I)$ is polynomial in the size of \I and \ch.
\end{lemma}
\begin{proof}
	In order to compute $\Apch(\I)$, we need to evaluate a fixed number of \FO-formulas a polynomial number of times (with exponent in the nesting depth of \D). 
	As evaluating a fixed \FO formula in the context of a partial structure is polynomial, the result follows.
\end{proof}

\begin{thm}\label{thm:wfe}
For a finite structure \I, the task of computing the $\Apch$-well-founded model of \D in the lattice $\struclat{\I,\ch}$ is polynomial in the size of \I and \ch.
\end{thm}
\begin{proof}
	Calculating the well-founded model of an approximator can be done with a polynomial number of applications of the approximator. 
	Furthermore, Lemma \ref{lem:ApplyApch} guarantees that each of these applications is polynomial as well.
\end{proof}

\begin{thm}
	Model expansion for \clog is NP-complete.
\end{thm}
\begin{proof}
	After guessing a model and a \chfun, Theorem \ref{thm:wfe} guarantees that checking that this is the well-founded model is polynomial. 
	Lemma \ref{lem:posseffectsettheory} shows that checking whether \D succeeds is polynomial as well.  Thus, model expansion is in NP.
	
	NP-hardness follows from the fact that model expansion for inductive definitions is NP-hard  and inductive definitions are shown to be a subclass of \clog theories, as argued in Section \ref{ssec:foid}.
\end{proof}


\begin{example}\label{ex:SAT}
	We show how the SAT-problem can be encoded as model checking for \clog. 
	Consider a vocabulary $\voc^{SAT}_{IN}$ with unary predicates \pred{Cl} and  \pred{PS} and with  binary predicates \pred{Pos} and \pred{Neg}.
	Every SAT-problem can be encoded as a  $\voc^{SAT}_{IN}$-structure: $\pred{Cl}$ and $\pred{PS}$ are interpreted as the sets of clauses and propositional symbols respectively, $\pred{Pos}(c,p)$ (respectively $\pred{Neg}(c,p)$) holds if clause $c$ contains the literal $p$ (respectively $\lnot p$).

%
	
	We now extend $\voc^{SAT}_{IN}$ to a vocabulary $\voc^{SAT}_{ALL}$ with unary  predicates  \pred{Tr} and \pred{Fa} and a propositional symbol $\pred{Sol}$. 
	\pred{Tr} and \pred{Fa} encode an assignment of values (true or false) to propositional symbols, $\pred{Sol}$ means that the encoded assignment is a solution to the SAT problem. 
	Let $\D_{SAT}$ be the following causal theory:
	\begin{align*}
		&\Call{p}{\pred{PS}(p)}{\Cor{\pred{Tr}(p)}{\pred{Fa}(p)}}\\
		& \Cif{\forall c[\pred{Cl}(c)]: \exists p: \\ &\qquad(\pred{Pos}(c,p)\land \pred{Tr}(p) \lor  (\pred{Neg}(c,p)\land \pred{Fa}(p))}{\pred{Sol}}
	\end{align*}
	The first rules guesses an assignment.  
	The second rule says that \pred{Sol} holds if every clause has at least one true literal. 
	Model expansion of that theory with a structure interpreting $\voc^{SAT}_{IN}$ according to a SAT problem and interpreting \pred{Sol} as true, is equivalent with solving that SAT problem, hence model expansion is NP-hard (which we already knew).
%
	In order to show that model  \emph{checking} is NP-hard, we add the following \CEE to the theory $\D_{SAT}$.
	\begin{align*}
		&\Cif{\pred{Sol}}{(\Call{p}{\pred{PS}(p)}{\Cand{\pred{Tr}(p)}{\pred{Fa}(p)})}}
	\end{align*}

	Basically, this rules tells us to forget the assignment once we have derived that it is a model (i.e., we hide the witness of the NP problem). 
	Now, the original SAT problem has a solution if and only if the structure interpreting symbols in $\voc^{SAT}_{IN}$ according to a SAT problem and interpreting all other symbols as constant true is a model of the extended theory. Hence:
\end{example}


\begin{thm}
Model checking for \clog is NP-complete. 
\end{thm}

%

Model checking might be a hard task but in certain cases (including for $\D_{SAT}$) endogenous model expansion is not.
The results in Theorem \ref{thm:wfe} can sometimes be used to generate models, if we have guarantees to end in a state where \D succeeds. 

\begin{thm}\label{thm:modelGenerationP}
 If \D is a total\footnote{A causal theory is \emph{total} if for every \chfun \ch, $w(\Apch)$ is two-valued, i.e., roughly, if it does not contain relevant loops over negation.}   causal theory without \Cnew and \Csel-expressions, endogenous model expansion is in P.
\end{thm}

Note that Theorem \ref{thm:modelGenerationP} does not contradict Example \ref{ex:SAT} since in that example, $Sol$ is interpreted as true in the input structure, i.e., the performed inference is not endogenous model expansion.
It is future work to generalise Theorem \ref{thm:modelGenerationP}, i.e., to research which are sufficient restrictions on \D such that model expansion is in P.

It is a well-known result in database theory that query languages combining recursion and object-creation are computationally complete~\cite{aw/AbiteboulHV95}; \foidplus can be seen as such a language.

\newcommand\wn{\m{\pred{while_{new}}}}
\begin{thm}
Unbounded querying can simulate the language \wn from \cite{aw/AbiteboulHV95}.
\end{thm}
\begin{proof}
We already showed that we can create the natural numbers in \clog. 
Once we have natural numbers and the successor function $\pred{Succ}$, we add one extra argument to every symbol (this argument represents time). Now, we encode the looping construct from  \pred{while_{new}} as follows. An expression of the form {\tt while P do s} corresponds to the \CEE: 
$\Call{t}{P(t)}{ C},$
where $C$ is the translation of the expression $s$.
An expression {\tt P = new Q} corresponds to a \CEE (where the variable $t$ should be bound by a surrounding \pred{while}).
\[\Call{\xxx, t'}{\pred{Succ}(t,t')}{\Cnew{y}{ P(\xxx,y,t')} \lrule Q(\xxx,t).}\qedhere\]
\end{proof}
Now, it follows immediately from \cite{aw/AbiteboulHV95} that
\begin{corollary}
For every decidable class $\mathcal{S}$ of finite structures closed under isomorphism, there exists a \D such that unbounded exogenous model generation returns true with input \I iff $\I\in\mathcal{S}$.
\end{corollary}

%
%

 \section{Conclusion}\label{sec:concl}
In this paper we presented several normal forms for \foclog. 
We showed that every \foclog theory can be transformed to a \voc-equivalent deterministic \foclog theory and 
to a \voc-equivalent \foclog theory in \nestnf or in \deff.
Furthermore, as  \foclog theories in \deff  correspond exactly to \foid, these transformations reduce inference for \foclog to \foid. 
We implemented a prototype of this above transformation, resulting in the first \foclog solver.
We also gave several complexity results for inference in \clog. 
All of these results are valuable from a theoretical point of view, as they help to characterise \foclog, but also from a practical point of view, as they provide more insight in \foclog.       
 \bibliography{krrlib}

\end{document}